\documentclass[11pt]{article}

\usepackage[utf8]{inputenc}
\usepackage[T1]{fontenc}
\usepackage{amsmath, amssymb, amsthm}
\usepackage{geometry}
\usepackage{booktabs}
\usepackage{tabularx}
\usepackage{enumitem}
\usepackage{hyperref}
\usepackage{tikz}
\usetikzlibrary{arrows.meta, positioning}

\setlist[itemize]{leftmargin=*}
\setlist[enumerate]{leftmargin=*}

\newtheorem{definition}{Definition}
\newtheorem{proposition}{Proposition}

\title{A Predicate-Based Model for Computation over State Spaces}

\author{
Jaime Alexander Jim\'enez Lozano \\
Independent Researcher / QDSV Research Group \\
Colombia \\
\texttt{jaimeajl@qruba.site}
\and
Sebasti\'an Jim\'enez Giraldo \\
Universidad de los Andes \\
QDSV Research Group \\
Colombia \\
\texttt{s.jimenezg2345@uniandes.edu.co}
}

\date{}

\begin{document}

\maketitle

\begin{abstract}
Many mainstream programming interfaces represent computation procedurally, as sequences of instructions, control-flow constructs, and explicit execution steps. However, several important classes of problems are more naturally described declaratively: one specifies the set of candidate states and the condition that makes a state valid.

This paper formalizes a predicate-based abstraction for computation over state spaces. A computational problem is represented by a state space $S$ and a predicate $C:S\to\{0,1\}$. Solutions are the states that satisfy the predicate, while execution is delegated to a realization strategy for evaluating, sampling, searching, or otherwise characterizing this solution set. We introduce a minimal semantic-preservation contract that distinguishes the problem specification from backend-specific evaluators and establishes when composed predicates preserve their meaning across realizations.

The contribution is a unifying abstraction and preservation contract, rather than a new class of constraint problems or a claim that predicate evaluation is always efficient. Procedural algorithms, solvers, probabilistic methods, and quantum oracles are treated as possible realizations of the same semantic specification.

The model is related to constraint satisfaction, satisfiability, logic programming, relational query processing, model checking, high-level quantum languages, and quantum intermediate representations. Its relevance to quantum computation follows from the fact that a Boolean predicate can be materialized, when finite and efficiently representable, as a reversible or phase oracle over computational basis states. This makes the abstraction a bridge between declarative problem specification and quantum-oriented execution without requiring the problem itself to be stated as a circuit.
\end{abstract}

\section{Introduction}

A large part of practical computing is expressed operationally: a program specifies how a machine should transform input into output through a sequence of instructions. This operational view is essential and powerful, but it is not always the most natural representation of the problem being solved. In matching, assignment, reconciliation, ranking, scheduling, and many search problems, the user often knows the validity condition more directly than the procedure that should discover valid candidates.

For example, in financial reconciliation one can state that a bank transaction and an accounting record match if their amounts, dates, and references are sufficiently compatible. The condition is compact, while the procedural implementation may involve nested loops, indexing, pruning, ranking, and exception handling. Similarly, in assignment problems, the validity of a mapping from tasks to workers can be described by skill, capacity, and non-reuse constraints, while the search procedure may be classical, heuristic, probabilistic, or quantum-oriented.

This paper proposes an abstraction in which computation is represented as predicate evaluation over a state space. The problem is specified as a pair $(S,C)$, where $S$ is the set of candidate states and $C$ is a Boolean predicate over those states. The solution set is the subset of $S$ satisfying $C$. Execution is then a realization problem: one may enumerate, query, prune, sample, optimize, invoke a solver, or materialize an oracle without changing the semantic specification.

The paper does not claim that declarative computation is new. Constraint satisfaction problems (CSP), Boolean satisfiability (SAT), logic programming, relational databases, and model checking all provide established declarative or state-based perspectives. The claim is narrower: for certain classes of computational systems, it is useful to make the state-space predicate itself the primary computational object and to treat execution strategies as interchangeable realizations below that semantic layer.

This distinction is particularly important for quantum computing. Quantum programming frameworks often require the user to reason in terms of circuits, gates, or low-level primitives. Yet many quantum algorithms, including search and amplitude amplification, are naturally described in terms of marked states induced by a Boolean predicate. A predicate-based model preserves the problem definition above the circuit level; a circuit, when required by a backend, is a materialization of the predicate, not the problem itself.

\section{Contributions and Scope}

This work makes four contributions:

\begin{enumerate}
    \item It defines a compact formal model in which a problem is represented as a state space $S$ and a predicate $C:S\to\{0,1\}$.
    \item It defines predicate-faithful realization and a corresponding semantic-preservation contract between a specification and a backend evaluator.
    \item It establishes that faithful realizations preserve Boolean predicate composition and therefore preserve the solution set of composed predicates.
    \item It relates the abstraction to established declarative paradigms and to quantum oracle materialization, including reversible evaluation and phase marking.
\end{enumerate}

The scope is deliberately limited. The abstraction does not make every predicate efficient, imply a general quantum speedup, or replace existing solvers and quantum SDKs. Hardware may still require circuits; the model treats them as backend-specific realizations of a higher-level specification.

\section{Formal Model}

\subsection{State Spaces}

Let $D_1,\ldots,D_n$ be finite domains. A state space is the Cartesian product

\[
S = D_1 \times \cdots \times D_n.
\]

A state is a tuple

\[
s=(d_1,\ldots,d_n)\in S,
\]

where each $d_i\in D_i$. Countable domains may be used at the specification level, but executable instances require finite encodings, bounded restrictions, or other effective representations. This restriction is especially important for quantum and simulation backends, where states must be encoded in a finite register or finite data structure.

\begin{definition}[Predicate-based computational problem]
A predicate-based computational problem is a pair
\[
P=(S,C),
\]
where $S$ is a state space and $C:S\to\{0,1\}$ is a Boolean predicate. A state $s\in S$ is valid if and only if $C(s)=1$.
\end{definition}

The solution set is

\[
\operatorname{Sol}(P)=\{s\in S\mid C(s)=1\}.
\]

The predicate may encode constraints, scoring thresholds, matching conditions, feasibility rules, domain-specific relations, or compositions of simpler predicates.

\subsection{Predicate Composition}

Predicates can be composed using standard Boolean operators. Given predicates $C_1,C_2:S\to\{0,1\}$, one can define

\[
(C_1\land C_2)(s)=1 \iff C_1(s)=1 \text{ and } C_2(s)=1,
\]

\[
(C_1\lor C_2)(s)=1 \iff C_1(s)=1 \text{ or } C_2(s)=1,
\]

and

\[
(\neg C_1)(s)=1 \iff C_1(s)=0.
\]

Weighted or ranked decisions can also be represented by introducing a score function $F:S\to\mathbb{R}$ and a threshold $\theta$:

\[
C(s)=\mathbf{1}[F(s)\geq \theta].
\]

This allows ordinary decision models, ranking rules, and feasibility filters to be expressed as predicates while preserving the possibility of returning additional metadata such as scores, ranks, or confidence measures.

\subsection{Semantic Resolution}

The model separates three layers:

\begin{enumerate}
    \item \textbf{Specification}: the state space $S$ and predicate $C$.
    \item \textbf{Intermediate representation}: a structured representation of $S$, $C$, bindings, operators, and execution requirements.
    \item \textbf{Realization}: a concrete execution strategy, such as enumeration, solver-backed search, sampling, or quantum oracle materialization.
\end{enumerate}

The same predicate can therefore be evaluated by different backends without changing the declared problem. This separation is the central abstraction: the problem is not the circuit, the loop, or the solver invocation; those are realizations of the problem.

\subsection{Realization and Semantic Preservation}

Let $P=(S,C)$ be a predicate-based problem. A backend realization $R$ provides an evaluator

\[
\widehat{C}_R:S\to\{0,1\}
\]

together with execution evidence $E_R$, such as the backend identity, encoding, resource information, counts, probabilities, approximation bounds, or reliability indicators.

\begin{definition}[Predicate-faithful realization]
A realization $R$ is \emph{predicate-faithful} for $P=(S,C)$ if
\[
\forall s\in S,\qquad \widehat{C}_R(s)=C(s).
\]
Its realized solution set is
\[
\operatorname{Sol}_R(P)=\{s\in S\mid \widehat{C}_R(s)=1\}.
\]
\end{definition}

Exact evaluators can satisfy this condition directly. Approximate, sampled, noisy, or probabilistic realizations instead declare an evidence contract describing how their observations estimate $C$ or $\operatorname{Sol}(P)$. Such realizations should not be reported as exact merely because they share the same input specification.

\begin{proposition}[Preservation under Boolean composition]
Let $C_1,C_2:S\to\{0,1\}$ have predicate-faithful realizations $\widehat{C}_{1,R}$ and $\widehat{C}_{2,R}$. Then the pointwise realizations of $C_1\land C_2$, $C_1\lor C_2$, and $\neg C_1$ are predicate-faithful. Consequently, their realized solution sets equal the solution sets of the corresponding specifications.
\end{proposition}

\begin{proof}
For every $s\in S$, faithfulness gives $\widehat{C}_{1,R}(s)=C_1(s)$ and $\widehat{C}_{2,R}(s)=C_2(s)$. Applying the same Boolean operator to equal Boolean values preserves equality. Thus the composed evaluator agrees pointwise with the composed predicate. Equality of the solution sets follows directly from their definitions.
\end{proof}

This elementary result makes the portability claim testable: a change of backend preserves the declared problem only when the backend evaluator, or its explicitly stated approximation contract, preserves predicate meaning.

\section{Computation Semantics and Complexity}

Under the predicate-based model, execution is the process of obtaining information about $\operatorname{Sol}(P)$. Depending on the realization strategy, execution may return all satisfying states, one satisfying state, a sample, an approximate count, a ranking, a probability distribution, or evidence that no satisfying state was found under the chosen method.

For finite domains, exhaustive evaluation has cost

\[
O(|S|\cdot T_C),
\]

where $T_C$ is the cost of evaluating $C(s)$ for a single state. Since

\[
|S|=\prod_{i=1}^{n}|D_i|,
\]

the state space grows exponentially in the number of variables when the domain sizes are bounded below by a constant greater than one.

This observation prevents an overstatement of the abstraction. Predicate-based representation can make a problem clearer and more portable, but it does not remove inherent combinatorial complexity. Practical execution requires realization strategies such as:

\begin{itemize}
    \item exhaustive enumeration for small domains;
    \item indexing, pruning, and decomposition for structured domains;
    \item SAT, CSP, SMT, or database query engines when applicable;
    \item probabilistic sampling or approximate search;
    \item heuristic or optimization-based methods;
    \item quantum oracle materialization when a finite reversible representation is available.
\end{itemize}

\section{Examples}

\subsection{Pair Matching}

Let $R$ be a set of records. Candidate pairs are represented by

\[
S=R\times R.
\]

A matching predicate can be defined as

\[
C(i,p)=(i\neq p)\land (|\operatorname{amount}(i)-\operatorname{amount}(p)|\leq \tau).
\]

The solution set contains all record pairs satisfying the matching condition.

\subsection{Financial Transaction Reconciliation}

Consider two datasets: $A$, containing bank transactions, and $B$, containing accounting records. The state space is

\[
S=A\times B.
\]

A state is a candidate reconciliation pair

\[
s=(a,b),\quad a\in A,\ b\in B.
\]

A reconciliation predicate can be defined as

\[
\begin{aligned}
C(a,b)=&\ (|\operatorname{amount}(a)-\operatorname{amount}(b)|\leq \tau)\\
&\land (|\operatorname{date}(a)-\operatorname{date}(b)|\leq \delta)\\
&\land (\operatorname{reference}(a)=\operatorname{reference}(b)).
\end{aligned}
\]

The solution set is

\[
\{(a,b)\in A\times B\mid C(a,b)=1\}.
\]

This expresses reconciliation as selection over candidate pairs rather than as a procedural matching program.

\subsection{Task Assignment}

Let $T$ be a set of tasks and $W$ a set of workers. A complete assignment can be modeled as a function from tasks to workers:

\[
S=W^T.
\]

A state $s\in S$ maps each task $t\in T$ to a worker $s(t)\in W$. A feasibility predicate can be defined as

\[
\begin{aligned}
C(s)=&\ (\forall t\in T:\operatorname{skill}(s(t))\geq \operatorname{requirement}(t))\\
&\land (\forall w\in W:\operatorname{workload}_{s}(w)\leq \operatorname{capacity}(w)),
\end{aligned}
\]

where $\operatorname{workload}_{s}(w)$ is the workload assigned to worker $w$ by state $s$.

Again, the problem is stated as the set of valid assignments, independently of the search method used to find them.

\subsection{Subset Selection}

Let $x_i\in\{0,1\}$ indicate whether item $i$ is selected. The state space is

\[
S=\{0,1\}^n.
\]

A capacity predicate can be defined as

\[
C(x)=\mathbf{1}\left[\sum_{i=1}^{n} w_i x_i\leq W\right].
\]

Additional predicates can encode value thresholds, diversity constraints, or incompatibility relations.

\subsection{Scored Decisions and Ranking}

Let each state $s\in S$ have prepared features $f_1(s),\ldots,f_m(s)$ and nonnegative weights $\alpha_1,\ldots,\alpha_m$ satisfying $\sum_{j=1}^{m}\alpha_j>0$. A score can be defined as

\[
F(s)=\frac{\sum_{j=1}^{m}\alpha_j f_j(s)}{\sum_{j=1}^{m}\alpha_j}.
\]

A decision predicate is then

\[
C(s)=\mathbf{1}[F(s)\geq \theta].
\]

The predicate determines feasibility, while $F(s)$ can be preserved for ranking valid states. This is useful in applications where one needs not only to decide whether a candidate is valid but also to prioritize the best candidates.

\section{Comparison with Algorithm-Centric Representation}

Both algorithm-centric and predicate-based representations have semantics. The distinction is where the semantics is made primary. In an algorithm-centric representation, the problem is encoded into a sequence of operations. In the predicate-based representation, the problem remains explicit as a condition over candidate states, and execution is delegated to a realization strategy.

\begin{center}
\small
\begin{tabularx}{\textwidth}{@{}p{0.23\textwidth}X X@{}}
\toprule
Aspect & Algorithm-centric representation & Predicate-based representation \\
\midrule
Primary object & Procedure or program & Predicate over state space \\
Problem meaning & Encoded in control flow and data structures & Explicit as $C:S\to\{0,1\}$ \\
Solution meaning & Output of execution & State satisfying $C$ \\
Execution form & Loops, branches, calls, instructions & Evaluation, search, sampling, realization \\
Portability & Depends on implementation & Same predicate, multiple realizations \\
Verification & Program behavior analysis & Check $C(s)$ and execution evidence \\
Quantum mapping & Circuit must be designed directly & Predicate may induce an oracle \\
\bottomrule
\end{tabularx}
\end{center}

The comparison is not a rejection of procedural computation. It identifies a useful abstraction boundary: the predicate is the model; the procedure is one possible implementation.

\section{Prototype Surface Language: QIntent}

A prototype realization of this abstraction can expose a restricted declarative surface language. In the QDSV prototype, this surface is called QIntent. QIntent is Python-like in syntax, but it is not general-purpose Python: it is a constrained language for expressing domains, fields, predicates, ranking, and selected domain-specific operations that can be compiled into a structured problem representation.

For tabular data, a row-oriented intent may be written conceptually as:

\begin{verbatim}
find_rows("candidate_index")
  .where("amount_score", ">=", 900)
  .rank_by("reference_score")
  .top_k(20)
\end{verbatim}

For a finite domain, a predicate may be written as:

\begin{verbatim}
x = domain(0, 15)
find(x).where((x * 3) == 9)
\end{verbatim}

For a combinatorial matching problem over rows:

\begin{verbatim}
i = domain(0, 3)
p = domain(0, 3)
find(i, p).where(
    abs(field(i, "amount") - field(p, "amount")) <= 5,
    no_reuse=True
)
\end{verbatim}

The purpose of such a language is not to replace mathematical formalization, solvers, or quantum SDKs. Its role is to preserve predicate-level intention while allowing the runtime to choose an appropriate realization strategy. A working prototype of QDSV and QIntent has been implemented, and preliminary evaluations have exercised the same problem specifications across multiple execution backends. The architecture and experimental evaluation are reserved for subsequent publications; the present paper focuses on the foundational abstraction and preservation contract.

\section{Quantum Interpretation}

The predicate-based abstraction has a natural interpretation in quantum computation when the state space is finite and admits an injective encoding $\operatorname{enc}:S\to\{0,1\}^{q}$. Writing $|s\rangle$ as shorthand for $|\operatorname{enc}(s)\rangle$, and assuming that the encoded state space can be prepared, a uniform superposition over $S$ is

\[
|\psi\rangle=\frac{1}{\sqrt{|S|}}\sum_{s\in S}|s\rangle.
\]

A Boolean predicate $C:S\to\{0,1\}$ can be represented reversibly by an oracle

\[
U_C|s\rangle|b\rangle=|s\rangle|b\oplus C(s)\rangle,
\]

where $b$ is an ancilla bit and $\oplus$ denotes addition modulo two. It can also induce a phase oracle

\[
O_C|s\rangle=(-1)^{C(s)}|s\rangle.
\]

States satisfying $C(s)=1$ are marked states. Algorithms such as Grover search and amplitude amplification exploit this structure by increasing the probability mass on marked states under suitable conditions.

Basis states outside $\operatorname{enc}(S)$, when the encoding does not fill the register, must be excluded during state preparation or explicitly marked invalid by the realized predicate. This mapping is structural, not automatic evidence of advantage. Efficient quantum realization requires an efficient reversible implementation of $C$, an appropriate state encoding and preparation procedure, sufficient coherence, and a backend capable of executing the resulting circuit or equivalent oracle representation.

The key abstraction is that the user need not define the problem as a circuit. The predicate is the semantic object. If a backend requires a circuit, the circuit is a materialization of the predicate for that backend.

\section{Related Work}

The proposed abstraction is connected to several established areas.

\textbf{Constraint satisfaction and satisfiability.} CSP and SAT represent problems through variables and constraints, with solutions corresponding to satisfying assignments. Classic work on consistency in networks of relations and satisfiability laid foundations for much of modern constraint and combinatorial reasoning \cite{mackworth1977,cook1971,karp1972,apt2003}. The present model does not replace these frameworks. It abstracts their common structure into a general state-space predicate view and emphasizes realization portability.

\textbf{Logic programming and declarative programming.} Logic programming separates logical specification from control, a distinction famously summarized as algorithm equals logic plus control \cite{kowalski1979,lloyd1984}. The predicate-based model is aligned with this idea, but uses the state space and predicate $C:S\to\{0,1\}$ as the primary computational object rather than centering on proof search.

\textbf{Relational query processing.} Relational databases define queries declaratively over relations, and query optimizers choose execution plans independently of the user's logical query \cite{codd1970,ullman1988}. This is closely related to the separation between specification and realization. The present abstraction generalizes beyond relational algebra toward arbitrary state spaces, composed predicates, ranking, and quantum-oriented materialization.

\textbf{Model checking.} Model checking evaluates properties over state-transition systems and can be understood as reasoning about states satisfying temporal or logical properties \cite{clarke1981,clarke1999}. The proposed model is simpler and more general at the level considered here: it focuses on predicates over candidate states rather than temporal transition systems, while sharing the emphasis on explicit state spaces and verifiable properties.

\textbf{Quantum search and oracles.} Quantum algorithms frequently rely on oracles that mark valid states. Grover search and amplitude amplification provide canonical examples \cite{grover1996,brassard2002,nielsen2010}. The predicate-based abstraction directly exposes the object that such oracles evaluate, making it useful as a high-level representation above circuit construction.

\textbf{High-level quantum languages and intermediate representations.} Recent systems raise the abstraction level in different ways. Silq provides a high-level language with semantics and safe automatic uncomputation \cite{bichsel2020}; OpenQASM~3 expands circuit-level representation with control flow, timing, and multiple levels of specificity \cite{cross2022}; and QIR provides an LLVM-based intermediate layer intended to decouple source languages from target hardware. Recent work has formalized code-safety aspects of QIR and demonstrated cross-platform execution \cite{luo2023,wong2024}. Qmod captures quantum algorithmic intent while delegating implementation decisions to synthesis automation \cite{vax2025}. The present model is complementary: it places a problem-level state space and predicate before quantum-program or circuit representation and makes semantic preservation across realizations an explicit contract.

\section{Limitations}

The abstraction has practical and theoretical limitations.

First, the size of $S$ may be prohibitively large. A clear predicate does not make search easy by itself. Second, some predicates may be expensive to evaluate or difficult to encode reversibly. Third, finite encodings can distort the original problem if domains are continuous, noisy, ambiguous, or poorly prepared. Fourth, different realization strategies may return different forms of evidence: complete solution sets, approximate samples, probability distributions, rankings, or partial results. These differences must be surfaced rather than hidden.

For quantum execution, the principal limitation is oracle construction. A predicate must be converted into a reversible or phase-marking operation, and this conversion may require additional qubits, circuit depth, decomposition choices, and error-aware execution. On present-day noisy hardware, backend limitations and noise may dominate the theoretical structure of the predicate.

\section{Discussion}

The predicate-based view provides a common language for several kinds of computation: business rules, matching, scoring, combinatorial search, declarative queries, and quantum oracles. Its value lies less in replacing existing methods and more in creating an abstraction layer above them.

This is useful when users can express what constitutes a valid or preferred state more naturally than they can express an algorithm. It is also useful when the same problem should be executed through different strategies. A predicate can be validated logically, tested on a simulator, executed through a classical or hybrid search method, or materialized into an oracle-like structure when quantum execution is available.

A practical system built on this abstraction should therefore expose evidence about realization: backend used, encoding choices, search strategy, circuit or oracle materialization when applicable, solution counts or estimates, probability distributions, ranking metadata, and confidence or reliability indicators. Without such evidence, the separation between predicate and realization may obscure important computational costs.

\section{Conclusion}

This paper introduced a predicate-based model for computation over state spaces. A problem is specified as a state space $S$ and a Boolean predicate $C:S\to\{0,1\}$; solutions are the states satisfying the predicate. The model separates problem semantics from execution strategy, allowing the same formal specification to be realized through enumeration, solvers, sampling, hybrid execution, or quantum oracle materialization.

The main contribution is a reframing: computation can be treated, for a broad class of problems, as state-space predicate evaluation rather than only as a procedural sequence of operations. This reframing is consistent with existing declarative traditions while providing a direct conceptual bridge to quantum computation, where predicates naturally correspond to marked states and oracle constructions.

Future work includes formalizing richer predicate and scoring languages, extending preservation contracts to approximate and probabilistic realizations, defining the intermediate representation, studying resource bounds, and reporting the prototype's evaluations on classical simulators and quantum hardware.

\section*{Author Contributions}

Jaime Alexander Jim\'enez Lozano conceived the predicate-based computational model, developed the main QDSV abstraction, and led the formalization and writing of the manuscript.

Sebasti\'an Jim\'enez Giraldo contributed to the conceptual analysis, review of the formal model, validation of examples, discussion of the computational interpretation, and manuscript revision.


\begin{thebibliography}{99}

\bibitem{codd1970}
E. F. Codd.
\newblock A relational model of data for large shared data banks.
\newblock \emph{Communications of the ACM}, 13(6):377--387, 1970.

\bibitem{cook1971}
S. A. Cook.
\newblock The complexity of theorem-proving procedures.
\newblock In \emph{Proceedings of the Third Annual ACM Symposium on Theory of Computing}, pages 151--158, 1971.

\bibitem{karp1972}
R. M. Karp.
\newblock Reducibility among combinatorial problems.
\newblock In R. E. Miller, J. W. Thatcher, and J. D. Bohlinger, editors, \emph{Complexity of Computer Computations}, pages 85--103. Plenum Press, 1972.

\bibitem{mackworth1977}
A. K. Mackworth.
\newblock Consistency in networks of relations.
\newblock \emph{Artificial Intelligence}, 8(1):99--118, 1977.

\bibitem{kowalski1979}
R. Kowalski.
\newblock Algorithm = logic + control.
\newblock \emph{Communications of the ACM}, 22(7):424--436, 1979.

\bibitem{lloyd1984}
J. W. Lloyd.
\newblock \emph{Foundations of Logic Programming}.
\newblock Springer, 1984.

\bibitem{ullman1988}
J. D. Ullman.
\newblock \emph{Principles of Database and Knowledge-Base Systems, Volume I}.
\newblock Computer Science Press, 1988.

\bibitem{clarke1981}
E. M. Clarke and E. A. Emerson.
\newblock Design and synthesis of synchronization skeletons using branching-time temporal logic.
\newblock In \emph{Logic of Programs}, Lecture Notes in Computer Science 131, pages 52--71. Springer, 1981.

\bibitem{clarke1999}
E. M. Clarke, O. Grumberg, and D. A. Peled.
\newblock \emph{Model Checking}.
\newblock MIT Press, 1999.

\bibitem{apt2003}
K. R. Apt.
\newblock \emph{Principles of Constraint Programming}.
\newblock Cambridge University Press, 2003.

\bibitem{grover1996}
L. K. Grover.
\newblock A fast quantum mechanical algorithm for database search.
\newblock In \emph{Proceedings of the Twenty-Eighth Annual ACM Symposium on Theory of Computing}, pages 212--219, 1996.

\bibitem{brassard2002}
G. Brassard, P. H\o yer, M. Mosca, and A. Tapp.
\newblock Quantum amplitude amplification and estimation.
\newblock \emph{Contemporary Mathematics}, 305:53--74, 2002.

\bibitem{nielsen2010}
M. A. Nielsen and I. L. Chuang.
\newblock \emph{Quantum Computation and Quantum Information: 10th Anniversary Edition}.
\newblock Cambridge University Press, 2010.

\bibitem{bichsel2020}
B. Bichsel, M. Baader, T. Gehr, and M. Vechev.
\newblock Silq: A high-level quantum language with safe uncomputation and intuitive semantics.
\newblock In \emph{Proceedings of the 41st ACM SIGPLAN Conference on Programming Language Design and Implementation}, pages 286--300, 2020.

\bibitem{cross2022}
A. W. Cross, A. Javadi-Abhari, T. Alexander, N. de Beaudrap, L. S. Bishop, S. Heidel, C. A. Ryan, P. Sivarajah, J. Smolin, J. M. Gambetta, and B. R. Johnson.
\newblock OpenQASM 3: A broader and deeper quantum assembly language.
\newblock \emph{ACM Transactions on Quantum Computing}, 3(3), Article 12, 2022.

\bibitem{luo2023}
J. Luo and J. Zhao.
\newblock Formalization of quantum intermediate representations for code safety.
\newblock arXiv:2303.14500, 2023. Subsequently published in the \emph{Journal of Systems and Software}, 219:112236, 2025.

\bibitem{wong2024}
E. Wong, V. Leyton-Ortega, D. Claudino, S. R. Johnson, A. J. Adams, S. Afrose, M. Gowrishankar, A. Cabrera, and T. S. Humble.
\newblock A cross-platform execution engine for the Quantum Intermediate Representation.
\newblock arXiv:2404.14299, 2024.

\bibitem{vax2025}
M. Vax, P. Emanuel, E. Cornfeld, I. Reichental, O. Opher, O. Roth, T. Michaeli, L. Preminger, L. Gazit, A. Naveh, and Y. Naveh.
\newblock Qmod: Expressive high-level quantum modeling.
\newblock arXiv:2502.19368, 2025.

\end{thebibliography}
\end{document}